\documentclass[letterpaper]{article}

\usepackage{microtype}
\usepackage[nolinenums]{custom_arxiv}
\usepackage[utf8]{inputenc}
\usepackage{times}
\usepackage{helvet}
\usepackage{courier}
\usepackage{url}
\usepackage{graphicx}
\usepackage{booktabs}
\usepackage{framed}
\usepackage{natbib}
\usepackage{amsmath}
\usepackage{amssymb}
\usepackage{amsthm}
\usepackage{float}
\usepackage{graphicx}
\usepackage{enumitem}
\usepackage{hyperref}

\usepackage{etoolbox}
\ifcsmacro{assumption}{}{
	
}
\ifcsmacro{corollary}{}{
	\newtheorem{corollary}{Corollary}[]
}
\ifcsmacro{definition}{}{
	\newtheorem{definition}{Definition}[]
}
\ifcsmacro{example}{}{
	
}
\ifcsmacro{fact}{}{
	
}
\ifcsmacro{informal_theorem}{}{
	
}
\ifcsmacro{lemma}{}{
	\newtheorem{lemma}{Lemma}[]
}
\ifcsmacro{proposition}{}{
	
}
\ifcsmacro{remark}{}{
	
}
\ifcsmacro{theorem}{}{
	\newtheorem{theorem}{Theorem}[]
}

\usepackage{mathtools}
\mathtoolsset{centercolon}
\newcommand{\defeq}{\mathrel{:\mkern-0.25mu=}}

\newcommand{\bbR}{\ensuremath{\mathbb{R}}}

\newcommand{\cX}{\ensuremath{\mathcal{X}}}
\newcommand{\cY}{\ensuremath{\mathcal{Y}}}

\newcommand{\cfrp}{CFR$^+$}

\newcommand{\be}{\begin{eqnarray}}
\newcommand{\ee}[1]{\label{#1}\end{eqnarray}}

	\newcommand{\ese}{\end{eqnarray*}}
	\newcommand{\bse}{\begin{eqnarray*}}
	\def\beq{\begin{equation}}
	\def\eeq{\end{equation}}

	\def\fnote#1{\footnote}

	\def\*{{{\LARGE\bf $^*$}}}

	\def\cJ{{\cal J}}
	\def\cK{{\cal K}}

	\def\cR{{\cal R}}

	\def\cX{{\cal X}}
	\def\cY{{\cal Y}}

	\def\argmin{\mathop{\rm argmin}}

\newcommand{\sprm}[3]{$(#1,#2,#3)$-stable-predictive}
\newcommand{\seqf}[1]{X^{\triangle}_{#1}}
\newcommand{\sfrm}[1]{\cR^{\triangle}_{#1}}
\newcommand{\sfrT}[1]{R^{\triangle,T}_{#1}}
\newcommand{\sfm}[2]{m^{\triangle,#1}_{#2}}
\newcommand{\sfell}[2]{\ell^{\triangle,#1}_{#2}}
\newcommand{\sfx}[2]{x^{\triangle,#1}_{#2}}
\newcommand{\localrm}[1]{\hat\cR_{#1}}
\newcommand{\nextv}[1]{\mathcal{C}_{#1}}
\newcommand{\subt}[1]{{\triangle}_{#1}}
\newcommand{\laminarregret}[2]{{\hat R}^{#2}_{#1}}

\usepackage{tikz}
\DeclareRobustCommand{\treeplexproduct}{%
  \tikz{%
    \draw (0,0) circle[radius=1.1mm];%
    \draw (-.55mm, -.55mm) -- (.55mm, .55mm);%
    \draw (-.55mm, .55mm) -- (.55mm, -.55mm);%
  }%
}

\DeclareRobustCommand{\treeplexproduct}{%
  \tikz[]{%
    \draw (0,0) circle[radius=1.1mm];%
    \draw (-.55mm, -.55mm) -- (.55mm, .55mm);%
    \draw (-.55mm, .55mm) -- (.55mm, -.55mm);%
  }%
}

\title{Stable-Predictive Optimistic Counterfactual\\Regret Minimization}

\author{
  Gabriele Farina\\
  Computer Science Department\\
  Carnegie Mellon University\\
  Pittsburgh, PA 15213 \\
  \texttt{gfarina@cs.cmu.edu}\\
   \And
  Christian Kroer\\
  IEOR Department\\
  Columbia University\\
  New York NY 10027\\
  \texttt{christian.kroer@columbia.edu}\\
  \AND
  Noam Brown\\
  Computer Science Department\\
  Carnegie Mellon University\\
  Pittsburgh, PA 15213 \\
  \texttt{noamb@cs.cmu.edu}\\
  \And
  Tuomas Sandholm\\
  Computer Science Department\\
  Carnegie Mellon University\\
  Pittsburgh, PA 15213 \\
  \texttt{sandholm@cs.cmu.edu}\\
}

\begin{document}
\maketitle
\begin{abstract}
  The CFR framework has been a powerful tool for solving large-scale extensive-form games in practice. However, the theoretical rate at which past CFR-based algorithms converge to the Nash equilibrium is on the order of $O(T^{-1/2})$, where $T$ is the number of iterations. In contrast, first-order methods can be used to achieve a $O(T^{-1})$ dependence on iterations, yet these methods have been less successful in practice. In this work we present the first CFR variant that breaks the square-root dependence on iterations. By combining and extending recent advances on predictive and stable regret minimizers for the matrix-game setting we show that it is possible to leverage ``optimistic'' regret minimizers to achieve a $O(T^{-3/4})$ convergence rate within CFR. This is achieved by introducing a new notion of stable-predictivity, and by setting the stability of each counterfactual regret minimizer relative to its location in the decision tree. Experiments show that this method is faster than the original CFR algorithm, although not as fast as newer variants, in spite of their worst-case $O(T^{-1/2})$ dependence on iterations.
\end{abstract}

\section{Introduction}

\emph{Counterfactual regret minimization} (CFR)~\citep{Zinkevich07:Regret} and later variants such as \emph{Monte-Carlo CFR}~\citep{Lanctot09:Montea}, \cfrp~\citep{Tammelin15:Solving}, and \emph{Discounted CFR}~\citep{Brown19:Solving}, have been the practical state-of-the-art in solving large-scale zero-sum \emph{extensive-form games} (EFGs) for the last decade. These algorithms were used as an essential ingredient for all recent milestones in the benchmark domain of poker~\citep{Bowling15:Heads,Moravvcik17:DeepStack,Brown17:Superhuman}. Despite this practical success all known CFR variants have a significant theoretical drawback: their worst-case convergence rate is on the order of $O({T}^{-1/2})$, where $T$ is the number of iterations. In contrast to this, there exist first-order methods that converge at a rate of $O(T^{-1})$~\citep{Hoda10:Smoothing,Kroer15:Faster,Kroer18:Faster}. However, these methods have been found to perform worse than newer CFR algorithms such as {\cfrp}, in spite of their theoretical advantage~\citep{Kroer18:Faster,Kroer18:Solving}.

In this paper we present the first CFR variant which breaks the square-root dependence on the number of iterations. By leveraging recent theoretical breakthroughs on ``optimistic'' regret minimizers for the matrix-game setting, we show how to set up optimistic counterfactual regret minimizers at each information set such that the overall algorithm retains the properties needed in order to accelerate convergence. In particular, this leads to a \emph{predictive} and \emph{stable} variant of CFR that converges at a rate of $O(T^{-3/4})$.

Typical analysis of regret-minimization leads to a convergence rate of $O(T^{-1/2})$ for solving zero-sum matrix games. However, by leveraging the idea of \emph{optimistic learning}~\citep{Chiang12:Online,Rakhlin13:Online,Rakhlin13:Optimization,Syrgkanis15:Fast,Wang18:Acceleration}, Rakhlin and Sridharan show in a series of papers that it is possible to converge at a rate of $O(T^{-1})$ when leveraging cancellations that occur due to the \emph{optimistic mirror descent} (OMD) algorithm~\citep{Rakhlin13:Online,Rakhlin13:Optimization}. \citet{Syrgkanis15:Fast} build on this idea, and introduce the \emph{optimistic follow-the-regularized-leader} (OFTRL) algorithm; they show that even when the players do not employ the same algorithm, a rate of $O(T^{-3/4})$ can be achieved as long as each algorithm belongs to a class of algorithms that satisfy a stability criterion and leverage predictability of loss inputs.  We build on this latter generalization. Because we can only perform the optimistic updates locally with respect to counterfactual regrets we cannot achieve the cancellations that leads to a rate of $O(T^{-1})$; instead we show that by carefully instantiating each counterfactual regret minimizer it is possible to maintain predictability and stability with respect to the overall decision-tree structure, thus leading to a convergence rate of $O(T^{-3/4})$. In order to achieve these results we introduce a new variant of stable-predictivity, and show that each local counterfactual regret minimizer must have its stability set relative to its location in the overall strategy space, with regret minimizers deeper in the decision tree requiring more stability.

In addition to our theoretical results we investigate the practical performance of our algorithm on several poker subgames from the \emph{Libratus} AI which beat top poker professionals~\citep{Brown17:Superhuman}. We find that our CFR variant coupled with the OFTRL algorithm and the entropy regularizer leads to better convergence rate than the vanilla CFR algorithm with regret matching, while it does  not outperform the newer state-of-the-art algorithm \emph{Discounted CFR} (DCFR)~\citep{Brown19:Solving}. This latter fact is not too surprising, as it has repeatedly been observed that {\cfrp}, and the newer and faster DCFR, converges at a rate \emph{better} than $O(T^{-1})$ for many practical games of interest, in spite of the worst-case rate of $O(T^{-1/2})$.

The reader may wonder why we care about breaking the square-root barrier within the CFR framework. It is well-known that a convergence rate of $O(T^{-1})$ can be achieved outside the CFR framework. As mentioned previously, this can be done with first-order methods such as the \emph{excessive gap technique}~\citep{Nesterov05:Excessive} or \emph{mirror prox}~\citep{Nemirovski04:Prox} combined with a dilated distance-generating function~\citep{Hoda10:Smoothing,Kroer15:Faster,Kroer18:Faster}. Despite this, there has been repeated interest in optimistic regret minimization within the CFR framework, due to the strong practical performance of CFR algorithms. \citet{Burch17:Time} tries to implement CFR-like features in the context of $O(T^{-1})$ FOMs and regret minimizers, while \citet{Brown19:Solving} experimentally tries optimistic variants of regret minimizers in CFR. We stress that these prior results are only experimental; our results are the first to rigorously incorporate optimistic regret minimization in CFR, and the first to achieve a theoretical speedup.

\textbf{Notation}. Throughout the paper, we use the following notation when dealing with $\bbR^n$. We use $\langle x,y\rangle$ to denote the \emph{dot product} $x^{\!\top}\!y$ of two vectors $x$ and $y$. We assume that a pair of dual norms $\|\cdot\|,\|\cdot\|_\ast$ has been chosen. These norms need not be induced by inner products. Common examples of such norm pairs are the $\ell_2$ norm which is self dual, and the $\ell_1,\ell_\infty$ norms, which are are dual to each other. We will make explicit use of the 2-norm: $\|x\|_2 \defeq \sqrt{\langle x, x\rangle}$.

\section{Sequential Decision Making and EFG Strategy Spaces}\label{sec:sdp}
A sequential decision process can be thought of as a tree consisting of two types of nodes: \emph{decision nodes} and \emph{observation nodes}. The set of all decision nodes is denoted as $\cJ$, and the set of all observation nodes with $\cK$. At each decision node $j \in \cJ$, the agent chooses a strategy from the simplex $\Delta^{n_j}$ of all
probability distributions over the set $A_j$ of $n_j = |A_j|$ actions available at that decision node.
An action is sampled according to the chosen distribution, and the agent then waits to play again. While waiting, the agent might receive a signal (observation) from the process; this possibility is represented with an observation node. At a generic observation point $k \in \cK$, the agent might receive $n_k$ signals; the set of signals that the agent can observe is denoted as $S_k$.
The observation node that is reached by the agent after picking action $a \in A_j$ at decision point $j\in \cJ$ is denoted by $\rho(j, a)$. Likewise, the decision node reached by the agent after observing signal $s\in S_k$ at observation point $k \in \cK$ is denoted by $\rho(k, s)$. The set of all observation points reachable from $j \in \cJ$ is denoted as $\nextv{j} \defeq \{\rho(j, a): a\in A_j\}$. Similarly, the set of all decision points reachable from $k \in \cK$ is denoted as $\nextv{k} \defeq \{\rho(k, s): s\in S_k\}$. To ease the notation, sometimes we will use the notation $\nextv{ja}$ to mean $\nextv{\rho(j,a)}$.
A concrete example of a decision process is given in the next subsection.

At each decision point $j\in \cJ$ in a sequential decision process, the decision $\hat x_j \in \Delta^{n_j}$ of the agent incurs an (expected) linear loss $\langle \ell_j, \hat x_j\rangle$. The expected loss throughout the whole process is therefore
$
  \sum_{j\in\cJ} \pi_j \langle \ell_j, \hat x_j \rangle,
$
where $\pi_j$ is the probability of the agent reaching decision point $j$, defined as the product of the probability with which the agent plays each action on the path from the root of the process to $j$.

In extensive-form games where all players have \emph{perfect recall} (that is, they never forget about their past moves or their observations), all players face a sequential decision process. The loss vectors $\{\ell_j\}$ are defined based on the strategies of the opponent(s) as well as the chance player. However, as already observed by~\citet{Farina19:Online}, sequential decision processes are more general and can model other settings as well, such as POMDPs and MDPs when the decision maker conditions on the entire history of observations and actions.

\subsection{Example: Sequential Decision Process for the First Player in Kuhn Poker}
As an illustration, consider the game of Kuhn poker~\citep{Kuhn50:Simplified}.
Kuhn poker consists of a three-card deck: king, queen, and jack. Each player
first has to put a payment of 1 into the pot. Each player is then dealt one of the three cards, and the third is put
aside unseen. A single round of betting then occurs:
\begin{itemize}[nolistsep]
\item Player $1$ can check or bet $1$.
  \begin{itemize}[nolistsep]
  \item If Player $1$ checks Player $2$ can check or raise $1$.
    \begin{itemize}[nolistsep]
      \item If Player $2$ checks a showdown occurs.
      \item If Player $2$ raises Player $1$ can fold or call.
        \begin{itemize}
          \item If Player $1$ folds Player $2$ takes the pot.
          \item If Player $1$ calls a showdown occurs.
          \end{itemize}
        \end{itemize}
      \item If Player $1$ raises Player $2$ can fold or call.
        \begin{itemize}[nolistsep]
        \item If Player $2$ folds Player $1$ takes the pot.
        \item If Player $2$ calls a showdown occurs.
        \end{itemize}
      \end{itemize}
\end{itemize}
If no player has folded, a showdown occurs where the player with the higher card wins.
The sequential decision process the Player 1 is shown in Figure~\ref{fig:kuhn treeplex
  player1}, where \treeplexproduct{} denotes an observation point. In that example, we have: $\cJ = \{X_0,X_1,X_2,X_3,X_4,X_5,X_6\}$; $n_0 = 1$; $n_j = 2$ for all $j\in\cJ
\setminus \{X_0\}$; $A_{X_0} = \{\text{start}\}$, $A_{X_1} = A_{X_2} = A_{X_3} = \{\text{check},
\text{raise}\}$, $A_{X_4} = A_{X_5} = A_{X_6} = \{\text{fold}, \text{call}\}$;
$\nextv{\rho(X_0,\text{start})} = \{X_1, X_2, X_3\}$, $\nextv{\rho(X_1, \text{raise})}
= \emptyset$, $\nextv{\rho(X_3,\text{check})} = \{X_6\}$; etc.
\begin{figure}[ht]
  \centering\includegraphics[width=.45\linewidth]{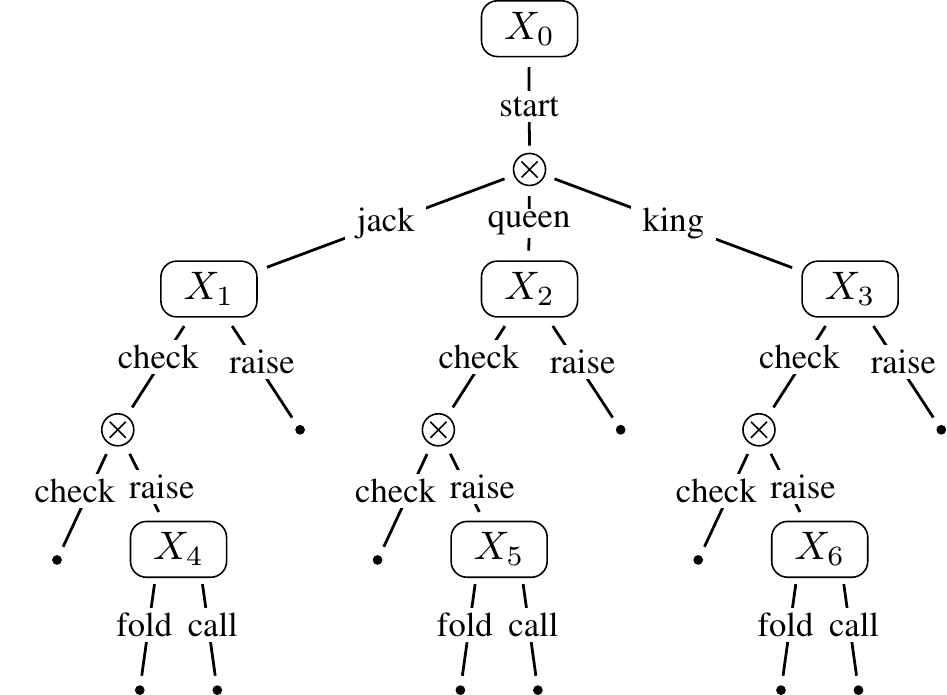}
   \caption{The sequential decision process for the first player in the game of Kuhn poker. \treeplexproduct{} denotes an observation point; small dots represents the end of the decision process.}
  \label{fig:kuhn treeplex player1}
\end{figure}

\subsection{Sequence Form for Sequential Decision Processes}\label{sec:sequence form}
The expected loss for a given strategy, as defined in Section~\ref{sec:sdp}, is non-linear in the vector of decisions variables $(\hat x_j)_{j\in\cJ}$. This non-linearity is due to the product $\pi_j$ of probabilities
of all actions on the path to from the root to $j$. We now present a well-known alternative representation of
this decision space which preserves linearity.

The alternative formulation is called the \emph{sequence form}. In the sequence-form
representation, the simplex strategy space at a generic decision point $j\in \cJ$ is scaled by the decision variable leading of the last action in the path from the root of the process to $j$. In this formulation, the value of a particular action
represents the probability of playing the whole \emph{sequence} of actions from
the root to that action. This allows each term in the expected loss to be
weighted only by the sequence ending in the corresponding action. The sequence
form has been used to instantiate linear programming~\citep{Stengel96:Efficient}
and first-order methods~\citep{Hoda10:Smoothing,Kroer15:Faster,Kroer18:Faster}
for computing Nash equilibria of zero-sum EFGs. There is a straightforward
mapping between a vector of decisions $(\hat x_j)_{j\in \cJ}$, one for each decision point, and its corresponding sequence form: simply assign
each sequence the product of probabilities in the sequence. We will let $\seqf{}$ denote the sequence-form representation of a vector of decisions $(\hat x_j)_{j\in \cJ}$. Likewise,
going from a sequence-form strategy $x^\triangle \in \seqf{}$ to a corresponding vector of decisions $(\hat x_j)_{j\in \cJ}$ can be done by dividing each
entry (sequence) in $x^\triangle$ by the value $x^\triangle_{p_j}$ where $p_j$ is the entry in $x^\triangle$ corresponding
to the unique last action that the agent took before reaching $j$.

Formally, the sequence-form representation $\seqf{}$ of a sequential decision process can be obtained recursively, as follows:
\begin{itemize}
  \item At every observation point $k \in \cK$, we let
    \begin{equation}\label{eq:sf observation point}
      \seqf{k} \defeq \seqf{j_1} \times \seqf{j_2} \times \dots \times \seqf{j_{n_k}},
    \end{equation}
    where $\{j_1, j_2, \dots, j_{n_k}\} = \nextv{k}$ are the children decision points of $k$.
  \item At every decision point $j \in \cJ$, we let
  \begin{equation}\label{eq:sf decision point}
    \seqf{j} \defeq \left\{\!\left(\begin{array}{c}
                               \lambda_1 \\
                               \vdots \\
                               \lambda_{n_j}\\
                               \hline
                               \lambda_1 x_{k_1}\\
                               \vdots \\
                               \lambda_{n_j} x_{k_{n_j}}
                             \end{array}\right): (\lambda_1, \dots, \lambda_n) \in \Delta^{n_j}, x_{k_1} \in \seqf{k_1}, x_{k_2}\in \seqf{k_2}, \dots, x_{k_{n_j}} \in \seqf{k_{n_j}} \right\}\!,
  \end{equation}
  where $\{k_1, k_2, \dots, k_{n_j}\} = \nextv{j}$ are the children observation points of $j$.
\end{itemize}
The sequence form strategy space for the whole sequential decision process is then $\seqf{r}$, where $r$ is the root of the process. Crucially, $\seqf{}$ is a convex and compact set, and the expected loss of the process is a linear function over $\seqf{}$.

With the sequence-form representation the problem of computing a Nash
equilibriun in an EFG can be formulated as a \emph{bilinear saddle-point
  problem} (BSPP). A BSPP has the form
\begin{equation}\label{eq:saddle point problem}
  \min_{x\in \cX} \max_{y\in \cY} x^{\!\top}\!\!A y,
\end{equation}
where $\cX$ and $\cY$ are convex and compact sets. In the case of extensive-form games, $\cX = \seqf{}$ and $\cY = Y^{\triangle}$ are the sequence-form strategy spaces of the sequential decision processes faced by the two players,
and $A$ is a sparse matrix encoding the leaf payoffs of the game.

\subsection{Notation when dealing with the extensive form}\label{sec:sf notation}
In the rest of the paper, we will make heavy use of the sequence form and its inductive construction given in~\eqref{eq:sfx k} and~\eqref{eq:sfx j}. We will consistently denote sequence-form strategies with a triangle superscript.
As we have already observed, vectors that pertain to the sequence-form have one entry for each sequence of the decision process, that is one entry for pair $(j, a)$ where $j\in\cJ, a \in A_j$. Sometimes, we will need to \emph{slice} a vector $v$ and isolate only those entries that refer to all decision points $j'$ and actions $a' \in A_{j'}$ that are at or below some $j\in \cJ$; we will denote such operation as $[v]_{\downarrow j}$. Similarly, we introduce the syntax $[v]_j$ to denote the subset of $n_j = |A_j|$ entries of $v$ that pertain to all actions $a \in A_j$ at decision point $j\in\cJ$.

\section{Stable-Predictive Regret Minimizers}

In this paper, we operate within the online learning framework called \emph{online convex optimization}~\citep{Zinkevich03:Online}. In particular, we restrict our attention to a modern subtopic: \emph{predictive} (also often called \emph{optimistic}) regret minimization~\citep{Chiang12:Online,Rakhlin13:Online,Rakhlin13:Optimization}.

As usual in this setting, a decision maker repeatedly plays against an unknown environment by making a sequence of decisions $x^1, x^2, \dots \in \cX \subseteq \mathbb{R}^n$, where the set $\cX$ of feasible decisions for the decision maker is convex and compact. The evaluation of the outcome of each decision $x^t$ is $\langle \ell^t\!,\, x^t\rangle$, where $\ell^t\in\cX$ is a convex \emph{loss vector}, unknown to the decision maker until after the decision is made. The peculiarity of \emph{predictive} regret minimization is that we also assume that the decision maker has access to \emph{predictions} $m^1, m^2, \dots$ of what the loss vectors $\ell^1, \ell^2, \dots$ will be. In summary, by \emph{predictive regret minimizer} we mean a device that supports the following two operations:
\begin{itemize}
  \item it provides the next decision $x^{t+1}\!\in\!\cX$ given a prediction $m^{t+1}$ of the next loss vector and
  \item it receives/observes the convex loss vectors $\ell^t$ used to evaluate decision $x^t$.
\end{itemize}

The learning is \emph{online} in the sense that the decision maker's (that is, device's) next decision, $x^{t+1}$, is based only on the previous decisions $x^1, \dots, x^t$, observed loss vectors $\ell^1, \dots, \ell^t$, and the prediction of the past loss vectors as well as the next one $m^1, \dots, m^{t+1}$.

Just as in the case of a regular (that is, non-predictive) regret minimizer, the quality metric for the predictive regret minimizer is its \emph{cumulative regret}, which is the difference between the loss cumulated by the sequence of decisions $x^1, \dots, x^T$ and the loss that would have been cumulated by playing the best-in-hindsight time-independent decision $\hat x$. Formally, the cumulative regret up to time $T$ is
\begin{equation}\label{eq:regret defn}
  R^T \defeq \sum_{t=1}^T \langle \ell^t\!,\, x^t\rangle  - \min_{\tilde x \in \cX} \left\{\sum_{t=1}^T \langle\ell^t\!,\,\tilde x\rangle\right\}\!.
\end{equation}%

We introduce a new class of predictive regret minimizers whose cumulative regret decomposes into a constant term plus a measure of the prediction quality, while maintaining stability in the sense that the iterates $x^1, \dots, x^T$ change slowly.

\begin{definition}[Stable-predictive regret minimizer]\label{def:sprm}
 A predictive regret minimizer is \sprm{\kappa}{\alpha}{\beta} if the following two conditions are met:
    \begin{itemize}
        \item \emph{Stability}. The decisions produced change slowly:
        \begin{equation}\label{eq:stability}
          \|x^{t+1} - x^t\| \le \kappa \quad\forall\, t \ge 1.
        \end{equation}
        \item \emph{Prediction bound}. For all $T$, the cumulative regret up to time $T$ is bounded according to
        \begin{equation}\label{eq:sp regret}
          R^T \le \frac{\alpha}{\kappa} + \beta\kappa\sum_{t=1}^T \|\ell^t - m^t\|_\ast^2.
        \end{equation}
        In other words, small prediction errors only minimally affect the regret accumulated by the device. If, in particular, the prediction $m^t$ matches the loss vector $\ell^t$ perfectly for all $t$, the cumulative regret remains asymptotically constant.
    \end{itemize}
\end{definition}

Our notion of stable-predictivity is similar to the \emph{Regret bounded
by Variation in Utilities} (RVU) property given by~\citet{Syrgkanis15:Fast}, which asserts that
\begin{equation}\label{eq:rvu}
 R^T \le \alpha' + \beta' \sum_{t=1}^T \|\ell^t - \ell^{t-1}\|_\ast^2 - \gamma' \sum_{t=1}^T \|x^t - x^{t-1}\|^2. \tag{RVU}
\end{equation}
However, there are several important differences:
\begin{itemize}
  \item \citet{Syrgkanis15:Fast} assume that $m^t = \ell^{t-1}$; this explains the term $\|\ell^t - \ell^{t-1}\|_\ast^2$ in~\eqref{eq:rvu} instead of $\|\ell^t - m^t\|_\ast^2$ in~\eqref{eq:sp regret}. One of the reason why we do not make assumptions on $m^t$ is that, unlike in matrix games, we will need to use modified predictions for each local regret minimizer, since we need to predict the local counterfactual loss.
  \item Our notion ignores the cancellation term $-\gamma' \sum\|x^t - x^{t-1}\|^2$; instead, we require the \emph{stabilty} property~\eqref{eq:stability}.
  \item The coefficients in the regret bound~\eqref{eq:sp regret} are forced to be inversely proportional, and tied to the choice of the stability parameter $\kappa$. \citet{Syrgkanis15:Fast} show that same correlation holds for the optimistic follow-the-regularized leader, but they don't require it in their definition of the RVU property.
\end{itemize}

\citet{Syrgkanis15:Fast} show that their optimistic follow-the-regularized-leader (OFTRL) algorithm, as well as the variant of the mirror descent algorithm presented by~\citet{Rakhlin13:Online}, satisfy~\eqref{eq:rvu}. In Section~\ref{sec:oftrl} we show that OFTRL also satisfies stable-predictivity.

\subsection{Relationship with Bilinear Saddle-Point Problems}
\label{sec:bspp}
In this subsection we show how stable-predictive regret minimization can be used to solve a  BSPP such as a Nash equilibrium problem in two-player  zero-sum extensive-form games with perfect recall (Sections~\ref{sec:sdp} and~\ref{sec:sequence form}).
The solutions of~\eqref{eq:saddle point problem} are called \emph{saddle points}. The \emph{saddle-point residual} (or \emph{gap}) $\xi$ of a point $(\bar x, \bar y)\in \cX \times \cY$, defined as
\begin{align*}
  \xi \defeq \max_{\hat{y}\in\cY}\, {\bar x}^{\!\top}\!A \hat y - \min_{\hat x\in\cX}\,{\hat x}^{\!\top}\!A \bar y,
\end{align*}%
measures how close $(\bar x, \bar y)$ is to being a saddle point (the lower the residual, the closer).

It is known that regular (non-predictive) regret minimization yields an anytime algorithm that produces a sequence of points $(\bar x^T, \bar y^T) \in \cX\times\cY$ whose residuals are $\xi^T = O(T^{-1/2})$. \citet{Syrgkanis15:Fast} observe that in the context of matrix games (i.e., when $\cX$ and $\cY$ are simplexes), RVU minimizers that also satisfy the stability condition~\eqref{eq:stability} can be used in place of regular regret minimizers to improve the convergence rate to $O(T^{-3/4})$. In what follows, we show how to extend the argument to stable-predictive regret minimizers and general bilinear saddle-point problems beyond Nash equilibria in two-player zero-sum matrix games.

A folk theorem explains the tight connections between low regret and low residual~\citep{Cesa06:Prediction}. Specifically, by setting up two regret minimizers (one for $\cX$ and one for $\cY$) that observe loss vectors given by
$
  \ell^t_\cX \defeq -Ay^{t},
  \ell^t_\cY \defeq A^{\!\top}\! x^{t},
$
the profile of average decisions
\begin{equation}\label{eq:avg decision}
  \left(\frac{1}{T}\sum_{t=1}^T x^t, \frac{1}{T}\sum_{t=1}^T y^t\right) \in \cX \times \cY
\end{equation}
has residual $\xi$ bounded from above according to
\[
  \xi \le \frac{1}{T}(R^T_\cX + R^T_\cY).
\]
Hence, by letting the predictions be defined as
$
  m^t_\cX \defeq \ell_\cX^{t-1},
  m^t_\cY \defeq \ell_\cY^{t-1},
$
and assuming that the predictive regret minimizers are \sprm{\kappa}{\alpha}{\beta}, we obtain that the residual $\xi$ of the average decisions~\eqref{eq:avg decision} satisfies
\begin{align*}
  T\xi &\le \frac{2\alpha}{\kappa} + \beta\kappa\sum_{t=1}^T \|{-A}y^t + Ay^{t-1}\|_\ast^2\\[-4mm]
        &\hspace{3.6cm}+ \beta\kappa\sum_{t=1}^T \|A^\top\! x^t - A^\top\! x^{t-1}\|_\ast^2\\
        &\le \frac{2\alpha}{\kappa} + \beta \|A\|_\text{op}^2 \kappa\!\left(
         \sum_{t=1}^{T} \|x^{t} \!-\! x^{t-1}\|^2 + \sum_{t=1}^{T} \|y^{t} \!-\! y^{t-1}\|^2\!\right)\\
        &\le \frac{2\alpha}{\kappa} + 2 \beta T \|A\|_\text{op}^2 \kappa^3,
\end{align*}
where the first inequality holds by~\eqref{eq:sp regret}, the second by noting that the operator norm $\|\cdot\|_\text{op}$ of a linear function is equal to the operator norm of its transpose, and the third inequality by the stability condition~\eqref{eq:stability}.
This shows that if the stability parameter $\kappa$ of the two stable-predictive regret minimizers is $\Theta(T^{-1/4})$, then the saddle point residual is $\xi = O(T^{-3/4})$, an improvement over the bound $\xi = O(T^{-1/2})$ obtained with regular (that is, non-predictive) regret minimizers. 
\subsection{Optimistic Follow the Regularized Leader}\label{sec:oftrl}
\emph{Optimistic follow-the-regularized-leader} (OFTRL) is a regret minimizer introduced by~\citet{Syrgkanis15:Fast}. At each time $t$, OFTRL outputs the decision
\begin{equation}\label{eq:oftrl choice}
  x^t = \argmin_{\tilde x\in\cX} \left\{\left\langle \tilde x, m^t +
  \sum_{t=1}^{T-1} \ell^t \right\rangle + \frac{1}{\eta}R(\tilde x)\right\}\!,
\end{equation}
where $\eta > 0$ is a free constant and $R(\cdot)$ is a 1-strongly convex regularizer with respect to the norm $\|\cdot\|$. Furthermore, let $\Delta_R \defeq \max_{x,y\in\cX} \{R(x) - R(y)\}$ denote the diameter of the range of $R$, and let $\Delta_\ell \defeq \max_t \max\{\|\ell^t\|_\ast, \|m^t\|_\ast\}$ be the maximum (dual) norm of any loss vector or prediction thereof.

A theorem similar to that of~\citet[Proposition~7]{Syrgkanis15:Fast}, which was obtained in the context of the RVU property, can be shown for the stable-predictive framework:
\begin{theorem}\label{thm:oftrl}
  OFTRL is a $3\Delta_\ell$\sprm{\eta}{\Delta_R}{1} regret minimizer.
\end{theorem}

We give a proof of Theorem~\ref{thm:oftrl} the appendix. When the loss vectors are further assumed to be non-negative, it can be shown that OFTRL is $2\Delta_\ell$\sprm{\eta}{\Delta_R}{1}, where we have substituted a factor of $2$ rather than the factor of $3$ in Theorem~\ref{thm:oftrl}.


\section{CFR as Regret Decomposition}\label{sec:cfr subtree}
In this section we offer some insights into CFR, and discuss what changes need to be made in order to leverage the power of predictive regret minimization. CFR is a framework for constructing a (non-predictive) regret minimizer $\cR^\triangle$ that operates over the sequence-form strategy space $\seqf{}$ of a sequential decision process. In accordance with Section~\ref{sec:sf notation}, we denote the decision produced by $\cR^\triangle$ at time $t$ as $\sfx{t}{}$; the corresponding loss functions is denoted as $\sfell{t}{}$.

One central idea in CFR is to define a localized notion of loss: for all $j\in\cJ$, CFR constructs the following linear \emph{counterfactual loss function} ${\hat \ell}^{t,\circ}_{j} : \Delta^{n_j} \to \bbR$. Intuitively, the counterfactual loss ${\hat \ell}^{t,\circ}_j( x_j)$ of a \emph{local} strategy $ x_j \in \Delta^{n_j}$ measures the loss that the agent would face were the agent allowed to change the strategy at decision point $j$ \emph{only}.
In particular, ${\hat \ell}^{t,\circ}_j( x_j)$ is the loss of an agent that follows the strategy $ x_j$ instead of $\sfx{t}{}$ at decision point $j$, but otherwise follows the strategy $\sfx{t}{}$ everywhere else.
Formally,
\begin{equation}\label{eq:counterfactual loss function}
  {\hat \ell}^{t,\circ}_j :  x_j = ( x_{ja_1}, \dots  x_{ja_{n_j}}) \mapsto  \langle [\sfell{t}{}]_j,  x_j \rangle + \sum_{a \in A_j} \left( x_{ja} \sum_{j' \in \nextv{ja}}\langle [\sfell{t}{}]_{\downarrow j'}, [\sfx{t}{}]_{\downarrow j'}\rangle\right).
\end{equation}

Since ${\hat\ell}^{t,\circ}_{j}$ is a linear function, it has a unique representation as a \emph{counterfactual loss vector} ${\hat\ell}^t_{j}$, defined as
\begin{equation}\label{eq:counterfactual loss vector}
  {\hat\ell}^{t,\circ}_{j}( x_j) = \langle {\hat\ell}^t_{j},  x_j \rangle \quad\forall\,  x_j \in \Delta^{n_j}.
\end{equation}
With this local notion of loss function, a corresponding local notion of regret for a sequence of decisions $\hat x_j^1, \dots, \hat x_j^T$, called the \emph{counterfactual regret}, is defined for each decision point $j\in \cJ$:
\[
  \laminarregret{j}{T} \defeq \sum_{t=1}^T \langle {\hat\ell}^t_{j}, \hat x^t_j \rangle - \min_{\tilde x_j\in \Delta^{n_j}} \sum_{t=1}^T \langle {\hat\ell}^t_j, \tilde x_j \rangle.
\]
Intuitively, $\laminarregret{j}{T}$ represents the difference between the loss that was suffered for picking $\hat x_j^t \in \Delta^{n_j}$ and the minimum loss that could be secured by choosing a different strategy \emph{at decision point $j$ only}. This is conceptually different from the definition of regret of $\cR^\triangle$, which instead measures the difference between the loss suffered and the best loss that could have been obtained, in hindsight, by picking \emph{any} strategy from the whole strategy space, with no extra constraints.

With this notion of regret, CFR instantiates one (non-stable-predictive) regret minimizer $\localrm{j}$ for each decision point $j\in\cJ$. Each local regret minimizer $\localrm{j}$ operates on the domain $\Delta^{n_j}$, that is, the space of strategies at decision point $j$ only. At each time $t$, $\cR^\triangle$ prescribes the strategy that, at each information set $j$, behaves according to the decision of $\localrm{j}$. Similarly, any loss vector $\sfell{t}{}$ input to $\cR^\triangle$ is processed as follows: (i) first, the counterfactual loss vectors $\{{\hat\ell}^t_{j}\}_{j\in \cJ}$, one for each decision point $j \in \cJ$, are computed; (ii) then, each $\localrm{j}$ observes its corresponding counterfactual loss vector ${\hat\ell}^t_j$.

Another way to look at CFR and counterfactual losses is as an inductive construction over subtrees. When a loss function relative to the whole sequential decision process is received by the root node, inductively each node of the sequential decision process does the following:
\begin{itemize}
  \item If the node receiving the loss vector is an observation node, the incoming loss vector is partitioned and forwarded to each child decision node. The partition of the loss vector is done so as to ensure that only entries relevant to each subtree are received down the tree.
  \item If the node receiving the loss vector is a decision node, the incoming loss vector is first forwarded as-is to each of the child observation points, and then it is used to construct the counterfactual loss vector ${\hat\ell}^t_{j}$ which is input into $\localrm{j}$.
\end{itemize}
This alternative point of view differs from the original one, but has been recently used by~Farina et al.~(\citeyear{Farina18:Composability,Farina19:Online}) to simplify the analysis of the algorithm.
When viewed from the above point of view, CFR is recursively building---in a bottom-up fashion---regret minimizers for each subtree starting from child subtrees.

In accordance with our convention (Section~\ref{sec:sf notation}), we denote $\sfrm{{v}}$, for $v\in\cJ\cup\cK$, the regret minimizer that operates on $\seqf{v}$ obtained by only considering the local regret minimizers in the subtree rooted at vertex $v$ of the sequential decision process. Analogously, we will denote with $\sfrT{v}$ the regret of $\sfrm{{v}}$ up to time $T$, and with $\sfell{t}{v}$ the loss function entering $\sfrm{{v}}$ at time $t$. In accordance with the above construction, we have that
\begin{equation}\label{eq:sf loss slice}
  \sfell{t}{k} = [\sfell{t}{j}]_{\downarrow k} \quad\forall k \in\nextv{j},\quad\text{and}\quad \sfell{t}{j} = [\sfell{t}{k}]_{\downarrow j} \quad \forall j \in \nextv{k}.
\end{equation}
Finally, we denote the decisions produced by $\sfrm{v}$ at time $t$ as $\sfx{t}{v}$. As per our discussion above, the decisions produced by $\cR^\triangle$ are tied together inductively according to
\begin{equation}\label{eq:sfx k}
  \forall k\in \cK,\quad \sfx{t}{k} = (\sfx{t}{j_1},\dots, \sfx{t}{j_{n_k}}), \quad\text{where }\{j_1, \dots, j_{n_k}\} = \nextv{k},
\end{equation}
and
\begin{equation}\label{eq:sfx j}
  \forall j\in \cJ,\quad \sfx{t}{j} = \left(\hat x^t_j, \hat x^t_{ja_1}\sfx{t}{\rho(j,a_1)}, \dots, \hat x^t_{ja_{n_j}}\sfx{t}{\rho(j,a_{n_j})}\right)\quad\text{where }\{a_1, \dots, a_{n_j}\} = A_j.
\end{equation}
The following two lemmas can be easily extracted from~\citet{Farina18:Composability}:

\begin{lemma}\label{lem:regret bound cartesian product}
  Let $k\in\cK$ be an observation node. Then,
  $\displaystyle
    \sfrT{k} = \sum_{j\in\nextv{k}} \sfrT{j}.
  $
\end{lemma}
\begin{proof}
  By definition of $\sfrT{k}$,
  \[
    \sfrT{k} = \sum_{t=1}^T \langle \sfell{t}{k}, \sfx{t}{k} \rangle - \min_{\tilde x^\triangle_k \in \seqf{k}} \sum_{t=1}^T \langle \sfell{t}{k}, \tilde x^\triangle_k\rangle.
  \]
  By using~\eqref{eq:sfx k} and~\eqref{eq:sf loss slice}, we can break the dot products and the minimization problem into independent parts, one for each $j\in\nextv{k}$:
  \begin{align*}
    \sfrT{k} &= \sum_{j\in \nextv{k}} \sum_{t=1}^T \langle \sfell{t}{j}, \sfx{t}{j} \rangle - \sum_{j\in\nextv{k}}\min_{\tilde x^\triangle_j \in \seqf{j}} \sum_{t=1}^T  \langle \sfell{t}{j}, \tilde x^\triangle_{j} \rangle\\
             &= \sum_{j\in \nextv{k}} \left(\sum_{t=1}^T \langle \sfell{t}{j}, \sfx{t}{j} \rangle - \min_{\tilde x^\triangle_j \in \seqf{j}} \sum_{t=1}^T  \langle \sfell{t}{j}, \tilde x^\triangle_{j} \rangle\right)\\
             &= \sum_{j\in \nextv{k}}\sfrT{j},
  \end{align*}
  as we wanted to show.
\end{proof}

\begin{lemma}\label{lem:regret bound ch}
  Let $j\in\cJ$ be a decision point. Then,
  $\displaystyle
    \sfrT{j} \le \hat R_j^T + \max_{k\in\nextv{j}} \sfrT{k}.
  $
\end{lemma}
\begin{proof}
  By definition of $\sfrT{j}$,
  \[
    \sfrT{j} = \sum_{t=1}^T \langle \sfell{t}{j}, \sfx{t}{j} \rangle - \min_{\tilde x^\triangle_j \in \seqf{j}} \sum_{t=1}^T \langle \sfell{t}{j}, \tilde x^\triangle_j\rangle.
  \]
  By combining~\eqref{eq:sfx j} and~\eqref{eq:sf loss slice}, we can break the dot products and the minimization problem into independent parts, one for each $k\in\nextv{j}$, as well as a part that depends solely on $\hat x_j$:
  \begin{align*}
    \sfrT{j} &=
                 \sum_{t=1}^T\left(
                    \langle [\sfell{t}{j}]_j, \hat x_j^t\rangle
                    + \smashoperator{\sum_{\substack{a\in A_j\\k=\rho(j,a)}}} \hat x^t_{ja}\langle \sfell{t}{k}, \sfx{t}{k} \rangle
                 \right)\\
                 &\hspace{2cm}
                 -\min_{\tilde x_j \in \Delta^{n_j}}\!\left\{\!
                    \left(\sum_{t=1}^T \langle [\sfell{t}{j}]_j, \tilde x_j\rangle\right)
                    +\smashoperator{\sum_{\substack{a\in A_j\\k=\rho(j,a)}}}  \tilde x_{ja} \left(\min_{\tilde x^\triangle_k \in \seqf{k}} \sum_{t=1}^T  \langle \sfell{t}{k}, \tilde x^\triangle_{k} \rangle\right)\!\!
                 \right\}\\
             &=
                \sum_{t=1}^T\left(
                    \langle [\sfell{t}{j}]_j, \hat x_j^t\rangle
                    + \smashoperator{\sum_{\substack{a\in A_j\\k=\rho(j,a)}}}  \hat x^t_{ja}\langle \sfell{t}{k}, \sfx{t}{k} \rangle
                 \right)\\
                 &\hspace{2cm}
                 -\min_{\tilde x_j \in \Delta^{n_j}}\!\left\{\!
                    \left(\sum_{t=1}^T \langle [\sfell{t}{j}]_j, \tilde x_j\rangle\right)
                    +\smashoperator{\sum_{\substack{a\in A_j\\k=\rho(j,a)}}}  \tilde x_{ja} \left(-\sfrT{k} + \sum_{t=1}^T \langle \sfell{t}{k}, \sfx{t}{k} \rangle\right)\!\!
                 \right\} \displaybreak \\
             &\le
                \sum_{t=1}^T\left(
                    \langle [\sfell{t}{j}]_j, \hat x_j^t\rangle
                    + \smashoperator{\sum_{\substack{a\in A_j\\k=\rho(j,a)}}}  \hat x^t_{ja}\langle \sfell{t}{k}, \sfx{t}{k} \rangle
                 \right)\\
                 &\hspace{1.8cm}
                 -\min_{\tilde x_j \in \Delta^{n_j}}\!\left\{\!
                    \sum_{t=1}^T \left(\langle [\sfell{t}{j}]_j, \tilde x_j\rangle
                    +\smashoperator{\sum_{\substack{a\in A_j\\k=\rho(j,a)}}}  \tilde x_{ja} \langle \sfell{t}{k}, \sfx{t}{k} \rangle\right)\!\!
                 \right\} + \max_{\tilde x_j \in \Delta^{n_j}} \sum_{a\in A_j} \tilde x_{ja} \sfrT{k},
  \end{align*}
  where the equality follows by the definition of $\sfrT{k}$, and the inequality follows from breaking the minimization of a sum into a sum of minimization problems. By identifying the difference between the first two terms as the counterfactual regret $\hat R^T_j$ (that is, the regret of $\localrm{j}$ up to time $T$), we obtain
  \begin{align*}
    \sfrT{j} &\le \hat R^T_j + \max_{\tilde x_j \in \Delta^{n_j}} \sum_{k\in\nextv{j}} \tilde x_{ja} \sfrT{k} = \hat R^T_j + \max_{k \in \nextv{j}} \sfrT{k},
  \end{align*}
  as we wanted to show.
\end{proof}

The two lemmas above do not make any assumption about the nature of the (localized) regret minimizers $\localrm{j}$, and therefore they are applicable even when the  $\localrm{j}$ are predictive or, specifically, stable-predictive.

\section{Stable-Predictive Counterfactual Regret Minimization}
Our proposed algorithm behaves exactly like CFR, with the notable difference that our local regret minimizers $\localrm{j}$ are stable-predictive and chosen to have specific stability parameters. Furthermore, the predictions $m_j^t$ for each local regret minimizer $\localrm{j}$ are chosen so as to leverage the predictivity property of the regret minimizers. Given a desired value of $\kappa^* > 0$, by choosing the stability parameters and predictions as we will detail later, we can guarantee that $\cR^\triangle$ is a \sprm{\kappa^*}{O(1)}{O(1)} regret minimizer.\footnote{Throughout the paper, our asymptotic notation is always with respect to the number of iterations $T$.}

\subsection{Choice of Stability Parameters}

We use the following scheme to pick the stability parameter of $\localrm{j}$. First, we associate a scalar $\gamma_v$ to each node $v\in\cJ \cup \cK$ of the sequential decision process. The value $\gamma_r$ of the root decision node is set to $\kappa^*$, and the value for each other node $v$ is set relative to the value $\gamma_u$ of their parent
\begin{equation}\label{eq:choice of gamma}
  \gamma_v \defeq \begin{cases}
    \displaystyle\frac{\gamma_u}{2\sqrt{n_u}} & \text{if } u\in \cJ\\[4mm]
    \displaystyle\frac{\gamma_u}{\sqrt{n_u}} & \text{if } u\in \cK.
  \end{cases}
\end{equation}
The stability parameter of each decision point $j\in \cJ$ is chosen according to
\begin{equation}\label{eq:choice of kappa}
  \kappa_j \defeq \frac{\gamma_j}{2\sqrt{n_j} B_j},
\end{equation}
where $B_j$ is an upper bound on the 2-norm of any vector in $\seqf{j}$. A suitable value of $B_j$ can be found by recursively using the following rules:
\begin{align}
  \forall k\in \cK,\quad B_k &= \sqrt{\sum_{j\in\nextv{k}} B_j^2} \nonumber\\
  \forall j\in \cJ,\quad B_j &= \sqrt{1 + \max_{k\in\nextv{j}} B^2_k} \label{eq:B_j}
\end{align}
At each decision point $j$, any stable-predictive regret minimizer that is able to guarantee the above stability parameter can be used. For example, one can use OFTRL where the stepsize $\eta$ is chosen appropriately. For example, assuming without loss of generality that all loss vectors involved have (dual) norm bounded by $1/3$, we can simply set the stepsize $\eta$ of the local OFTRL regret minimizer $\localrm{j}$ at decision point $j$ to be $\eta = \kappa_j$.

\subsection{Prediction of Counterfactual Loss Vectors}

Let $\sfm{t}{}$ be the prediction received by $\cR^\triangle$, concerning the future loss vector $\sfell{t}{}$. We will show how to process the prediction and produce \emph{counterfactual} prediction vectors $\hat m_j^t$ (one for each decision point $j\in\cJ$) for each local stable-predictive regret minimizer $\localrm{j}$.

Following the construction of the counterfactual loss functions defined in~\eqref{eq:counterfactual loss function}, for each decision point $j\in\cJ$ we define the \emph{counterfactual prediction function} ${\hat m}^{t,\circ}_{j} : \Delta^{n_j} \to \bbR$ as
\begin{equation*}
  \hat m^{t,\circ}_j : \Delta^{n_j} \ni  x_j = ( x_{ja_1}, \dots, x_{ja_{n_j}}) \mapsto \langle [\sfm{t}{}]_j,  x_j\rangle + \sum_{a\in A_j} \left( x_{ja} \sum_{j' \in \nextv{ja}} \langle [\sfm{t}{}]_{\downarrow j'}, [\sfx{t}{}]_{\downarrow j'} \rangle\right).
\end{equation*}

\textbf{Observation.} \textit{It important to observe that the counterfactual prediction function $\hat m_j^t$ depends on the decisions produced at time $t$ in the subtree rooted at $j$. In other words, in order to construct the prediction for what loss $\localrm{j}$ will observe \emph{after} producing the decision $x_j^t$, we use the ``future'' decisions $x^t_{ja}$ from the subtrees below $j\in J$.}

Similarly to what is done for the counterfactual loss function, we define the \emph{counterfactual loss prediction vector} ${\hat m}^t_{j}$, as the (unique) vector in $\bbR^n_j$ such that
\begin{equation}\label{eq:counterfactual prediction}
  {\hat m}^{t,\circ}_{j}(x_j) = \langle {\hat m}^t_{j}, x_j \rangle \quad\forall\, x_j \in \Delta^{n_j}.
\end{equation}

\subsection{Proof of Correctness}

We will prove that our choice of stability parameters~\eqref{eq:choice of gamma} and (localized) counterfactual loss predictions~\eqref{eq:counterfactual prediction} guarantee that $\cR^\triangle$ is a \sprm{\kappa^*}{O(1)}{O(1)} regret minimizer. Our proof is by induction on the sequential decision process structure: we prove that our choices yield a \sprm{\gamma_v}{O(1)}{O(1)} regret minimizer in the sub-sequential decision process rooted at each possible node $v \in \cJ\cup\cK$.
For observation nodes $v\in\cK$ the inductive step is performed via Lemma~\ref{lem:induction observation}, while for decision nodes $v\in\cJ$ the inductive step is performed via Lemma~\ref{lem:induction decision}.


We will prove both Lemma~\ref{lem:induction observation} and Lemma~\ref{lem:induction decision} with respect to the 2-norm. This does not come at the cost of generality, since all norms are equivalent on finite-dimensional vector spaces, that is, for every choice of norm $\|\cdot\|$, there exist constants $m,M>0$ such that for all $x$, $m\|x\|\le \|x\|_2 \le M\|x\|$.

\begin{lemma}\label{lem:induction observation}
  Let $k \in \cK$ be an observation node, and assume that $\sfrm{j}$ is a \sprm{\gamma_j}{O(1)}{O(1)} regret minimizer over the sequence-form strategy space $\seqf{j}$ for each $j\in\nextv{k}$. Then, $\sfrm{k}$ is a \sprm{\gamma_k}{O(1)}{O(1)} regret minimizer over the sequence-form strategy space $\seqf{k}$.
\end{lemma}
\begin{proof}
  By hypothesis, for all $j \in \nextv{k}$ we have
  \begin{equation}\label{eq:io1}
    \sfrT{j} \le \frac{O(1)}{\gamma_j} + O(1)\gamma_j\sum_{t=1}^T \| \sfell{t}{j} -  \sfm{t}{j}\|_2^2
  \end{equation}
  and
  \begin{equation}\label{eq:io2}
    \| \sfx{t}{j} - \sfx{t-1}{j}\|_2 \le \gamma_j,
  \end{equation}
  where $\sfx{t}{j}$ is the decision output by $\sfrm{j}$ at time $t$.

  Substituting~\eqref{eq:io1} into the regret bound of Lemma~\ref{lem:regret bound cartesian product}:
  \begin{align}
    \sfrT{k} &\le O(1)\sum_{j\in\nextv{k}}\frac{1}{\gamma_j} + O(1)\sum_{j\in\nextv{k}}\sum_{t=1}^T \gamma_j \|\sfell{t}{j} - \sfm{t}{j}\|_2^2 \nonumber\\
      &\le O(1)\frac{n_k^{3/2}}{\gamma_k} + O(1)\frac{\gamma_k}{\sqrt{n_k}}\sum_{t=1}^T\sum_{j\in\nextv{k}} \| \sfell{t}{j} -  \sfm{t}{j}\|_2^2 \nonumber \\
      &= \frac{O(1)}{\gamma_k} + O(1)\gamma_k\sum_{t=1}^T \| \sfell{t}{k} -  \sfm{t}{k}\|_2^2\label{eq:io3}
  \end{align}
  where the second inequality comes from substituting the value $\gamma_j = \gamma_k/\sqrt{n_k}$ as per~\eqref{eq:choice of gamma}, and the equality comes from the fact that the $\sfell{t}{j}$ and $\sfm{t}{j}$ form a partition of the vectors $\sfell{t}{k}$ and $\sfm{t}{k}$, respectively.

  We now analyze the stability properties of $\sfrm{k}$:
  \begin{align*}
    \| \sfx{t}{k} -  \sfx{t-1}{k}\|_2 = \sqrt{\sum_{j\in\nextv{k}} \| \sfx{t}{j} -  \sfx{t-1}{j}\|^2_2} \le \sqrt{\sum_{j\in\nextv{k}} \gamma_j^2} = \gamma_k,
  \end{align*}
  where the first equality follows from~\eqref{eq:sf observation point}, the inequality holds by~\eqref{eq:io2} and the second equality holds by substituting the value $\gamma_j = \gamma_k/\sqrt{n_k}$ as per~\eqref{eq:choice of gamma}. This shows that $\sfrm{k}$ is $\gamma_k$-stable. Combining this with the predictivity bound~\eqref{eq:io3} above, we obtain the claim.
\end{proof}

\begin{lemma}\label{lem:induction decision}
  Let $j \in \cJ$ be a decision node, and assume that $\sfrm{k}$ is a \sprm{\gamma_k}{O(1)}{O(1)} regret minimizer over the sequence-form strategy space $\seqf{k}$ for each $k\in\nextv{j}$. Suppose further that $\localrm{j}$ is a \sprm{\kappa_j}{O(1)}{O(1)} regret minimizer over the simplex $\Delta^{n_j}$. Then, $\sfrm{j}$ is a \sprm{\gamma_k}{O(1)}{O(1)} regret minimizer over the sequence-form strategy space $\seqf{j}$.
\end{lemma}
\begin{proof}
    By hypothesis, for all $k \in \nextv{j}$ we have
  \begin{equation}\label{eq:id1}
    \sfrT{k} \le \frac{O(1)}{\gamma_k} + O(1)\gamma_k\sum_{t=1}^T \| \sfell{t}{k} -  \sfm{t}{k}\|_2^2
  \end{equation}
  and
  \begin{equation}\label{eq:id2}
    \|\sfx{t}{k} - \sfx{t-1}{k}\|_2 \le \gamma_k.
  \end{equation}
  We substitute~\eqref{eq:id1} into the regret bound of Lemma~\ref{lem:regret bound ch}. The key observation is that the loss vector---and their predictions---entering the subtree rooted at ${k}$ ($k\in\nextv{j}$) are simply forwarded from ${j}$; with this, we obtain:
  \begin{equation}\label{eq:id3}
    R_{\subt{j}}^T \le \hat R_j^T + \frac{O(1)}{\gamma_k} + O(1)\gamma_k\sum_{t=1}^T \| \sfell{t}{j} - \sfm{t}{j}\|_2^2.
  \end{equation}
  On the other hand, by hypothesis $\localrm{j}$ is a \sprm{\kappa_j}{O(1)}{O(1)} regret minimizer. Hence,
  \begin{align}
    \hat R_j^T &\le \frac{O(1)}{\kappa_j} + O(1)\kappa_j \sum_{t=1}^T \| \hat\ell_j^t -  \hat m_{j}^t\|_2^2 \nonumber\\
               &= \frac{O(1)}{\gamma_j} + O(1) \gamma_j \sum_{t=1}^T \| \sfell{t}{j} -  \sfm{t}{j}\|_2^2, \label{eq:id4}
  \end{align}
  where the equality comes from the definition of $\kappa_j$ (Equation~\eqref{eq:choice of kappa}) and the fact that
  \begin{align*}
    \| \hat\ell_j^t -  \hat m_{j}^t\|_2^2
    &\le \sum_{k\in\nextv{j}} \|\sfx{t}{k}\|_2^2\cdot \|\sfell{t}{k} - \sfm{t}{k}\|_2^2 \\
    &\le \|\sfell{t}{j} - \sfm{t}{j}\|_2^2 \sum_{k\in\nextv{j}} B^2_k\\
    &= O(1) \|\sfell{t}{j} - \sfm{t}{j}\|_2^2.
  \end{align*}
  By substituting~\eqref{eq:id4} into~\eqref{eq:id3} and noting that $\gamma_k = O(1)\gamma_j$, we obtain
  \[
    \sfrT{j} \le \frac{O(1)}{\gamma_j} + O(1)\gamma_j\sum_{t=1}^T \| \sfell{t}{j} - \sfm{t}{j}\|_2^2,
  \]
  which establishes the predictivity of $\sfrm{j}$.

  To conclude the proof, we show that $\sfrm{j}$ has stability parameter $\gamma_j$. To this end, note that by~\eqref{eq:sf decision point}
  \begin{align*}
      \| \sfx{t}{j} -  \sfx{t-1}{j}\|^2_2
           &=\left\|\left(\sum_{a\in A_j} {\hat x_{ja}^t  \sfx{t}{ja}}\right) - \left(\sum_{a\in A_j} \hat x_{ja}^{t-1} \sfx{t-1}{ja}\right) \right\|_2^2 + \|\hat x_j^t - \hat x_j^{t-1}\|_2^2\\
           &\le \|\hat x^t_j - \hat x^{t-1}_j\|_2^2 \left(1 + 2\sum_{k\in \nextv{j}} \|\sfx{t}{k}\|_2^2\right) + 2 \sum_{k\in\nextv{k}}\| \sfx{t}{k} - \sfx{t-1}{k}\|^2_2\\
           &\le 2 n_j B_j^2\|\hat x^t_j - \hat x^{t-1}_j\|_2^2 + 2\sum_{k\in\nextv{k}}\| \sfx{t}{k} -  \sfx{t-1}{k}\|^2_2,
  \end{align*}
  where we have used the Cauchy-Schwarz inequality and the definition of $B_j$ (Equation~\ref{eq:B_j}). By using the stability of $\localrm{j}$, that is $\|\hat x^t_j - \hat x^{t-1}_j\|_2^2 \le \kappa^2_j = \gamma^2_j/(4n_j B_j^2)$, as well as the hypothesis~\eqref{eq:id2} and~\eqref{eq:choice of gamma}:
  \begin{align*}
    \| \sfx{t}{j} -  \sfx{t-1}{j}\|_2 &\le \frac{\gamma_j^2}{2} + 2\sum_{k\in\nextv{j}} \left(\frac{\gamma_j}{2\sqrt{n_j}}\right)^2 = \frac{\gamma_j^2}{2} + 2n_j \left(\frac{\gamma_j}{2\sqrt{n_j}}\right)^2 = \gamma_j^2.\\
  \end{align*}
  Hence, $\sfrm{j}$ has stability parameter $\gamma_j$ as we wanted to show.
\end{proof}

Putting together Lemma~\ref{lem:induction observation} and Lemma~\ref{lem:induction decision}, and using induction on the sequential decision process structure, we obtain the following formal statement.

\begin{corollary}
  Let $\kappa^* > 0$. If:
  \begin{enumerate}
   \item Each localized regret minimizer $\localrm{j}$ is \sprm{\kappa_j}{O(1)}{O(1)} and produces decisions over the local (simplex) action space $\Delta^{n_j}$, where $\kappa_j$ is as in~\eqref{eq:choice of kappa}; and
   \item $\localrm{j}$ observes the counterfactual loss prediction $\hat m_j^t$ as defined in~\eqref{eq:counterfactual prediction}; and
   \item $\localrm{j}$ observes the counterfactual loss vectors $\hat \ell_j^t$ as defined in~\eqref{eq:counterfactual loss vector},
  \end{enumerate}
  then $\cR^\triangle$ is a \sprm{\kappa^*}{O(1)}{O(1)} regret minimizer that operates over the sequence-form strategy space $\tilde X$.
\end{corollary}

By combining the above result with the arguments of Section~\ref{sec:bspp}, we conclude that by constructing two \sprm{\Theta(T^{1/4})}{O(1)}{O(1)} regret minimizers, one per player, using the construction above, we obtain an algorithm that can approximate a Nash equilibrium and at time $T$ the average strategy produces an $O(T^{-3/4})$-Nash equilibrium in a two-player zero-sum game.

\section{Experiments}

Our techniques are evaluated in the benchmark domain of heads-up no-limit Texas hold'em poker (HUNL) subgames. In HUNL, there are two players $P_1$ and $P_2$ that each start the game with \$20,000. The position of the players switches after each hand. The players alternate taking turns and may choose to either fold, call, or raise on their turn. Folding results in the player losing and the money in the pot being awarded to the other player. Calling means the player places a number of chips in the pot equal to the opponent's share. Raising means the player adds more chips to the pot than the opponent's share. There are four betting rounds in the game. A round ends when both players have acted at least once and the most recent player has called. Players cannot raise beyond the \$20,000 they start with. All raises must be at least \$100 and at least as larger as any previous raise in that round.

At the start of the game $P_1$ must place \$100 in the pot and $P_2$ must place \$50 in the pot. Both players are then dealt two cards that only they observe from a 52-card deck. A round of betting then occurs starting with $P_2$. $P_1$ will be the first to act in all subsequent betting rounds. Upon completion of the first betting round, three \emph{community} cards are dealt face up. After the second betting round is over, another community card is dealt face up. Finally, after that betting round one more community card is revealed and a final betting round occurs. If no player has folded then the player with the best five-card poker hand, out of their two private cards and the five community cards wins the pot. The pot is split evenly if there is a tie.

The most competitive agents for HUNL solve portions of the game (referred to as subgames) in real time during play~\citep{Brown17:Safe,Moravvcik17:DeepStack,Brown17:Superhuman,Brown18:Depth}. For example, \emph{Libratus} solved in real time the remainder of HUNL starting on the third betting round. We conduct our experiments on four open-source subgames solved by \emph{Libratus} in real time during its competition against top humans in HUNL.\footnote{\url{https://github.com/CMU-EM/LibratusEndgames}} Following prior convention, we use the bet sizes of 0.5$\times$ the size of the pot, 1$\times$ the size of the pot, and the all-in bet for the first bet of each round. For subsequent bets in a round, we consider 1$\times$ the pot and the all-in bet.

Subgames 1 and 2 occur over the third and fourth betting round. Subgame 1 has \$500 in the pot at the start of the game while Subgame 2 has \$4,780.
Subgames 3 and 4 occur over only the fourth betting round. Subgame 1 has \$500 in the pot at the start of the game while Subgame 4 has \$3,750. We measure exploitability in terms of the standard metric: milli big blinds per game (mbb/g), which is the number of big blinds ($P_1$'s original contribution to the pot) lost per hand of poker multiplied by 1,000 and is the standard measurement of win rate in the related literature.

We compare the performance of three algorithms: vanilla CFR (i.e. CFR with regret matching; labeled CFR in plots), the current state-of-the-art algorithm in practice, Discounted CFR~\citep{Brown19:Solving} (labeled DCFR in plots), and our stable-predictive variant of CFR with OFTRL at each decision point (labeled OFTRL in plots). For OFTRL we use the stepsize that the theory suggests in our experiments on subgames 3 and 4 (labeled OFTRL theory). For subgames 1 and 2 we found that the theoretically-correct stepsize is much too conservative, so we also show results with a less-conservative parameter found through dividing the stepsize by 10, 100, and 1000, and picking the best among those (labeled OFTRL tuned). For all games we show two plots: one where all algorithms use simultaneous updates, as CFR traditionally uses, and one where all algorithms use alternating updates, a practical change that usually leads to better performance.

\begin{figure}[!h]
  \centering
  \includegraphics[width=.49\linewidth]{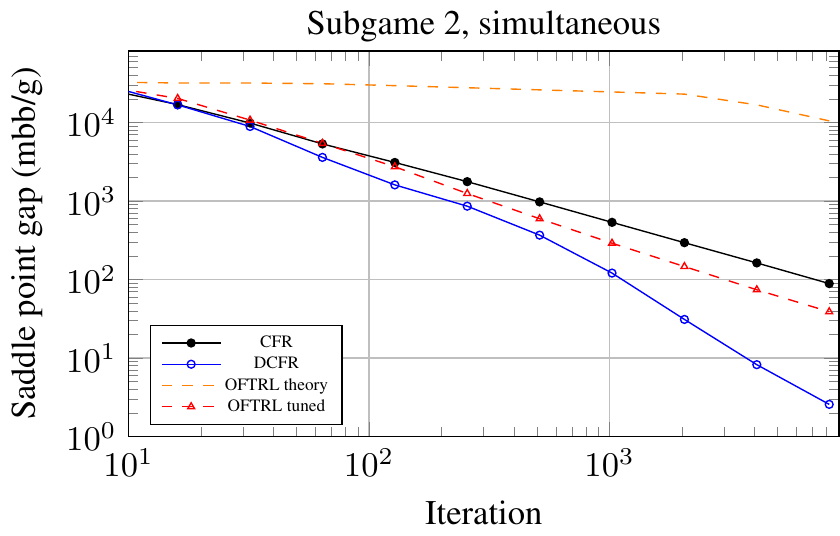}\hfill
  \includegraphics[width=.49\linewidth]{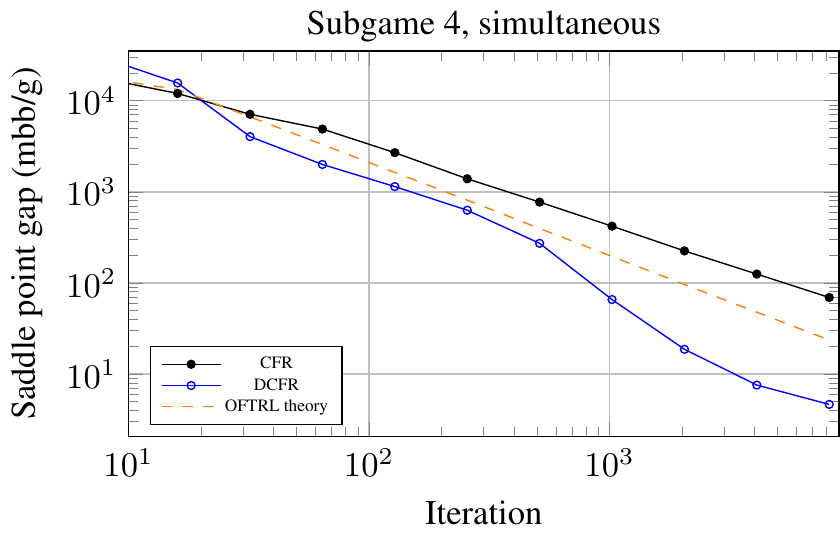}
  \caption{Convergence rate with iterations on the x-axis, and the exploitability in mbb. All algorithms use simultaneous updates.}
  \label{fig:plots simultaneous 2 4}
\end{figure}

\begin{figure}[H]
  \centering
  \includegraphics[width=.49\linewidth]{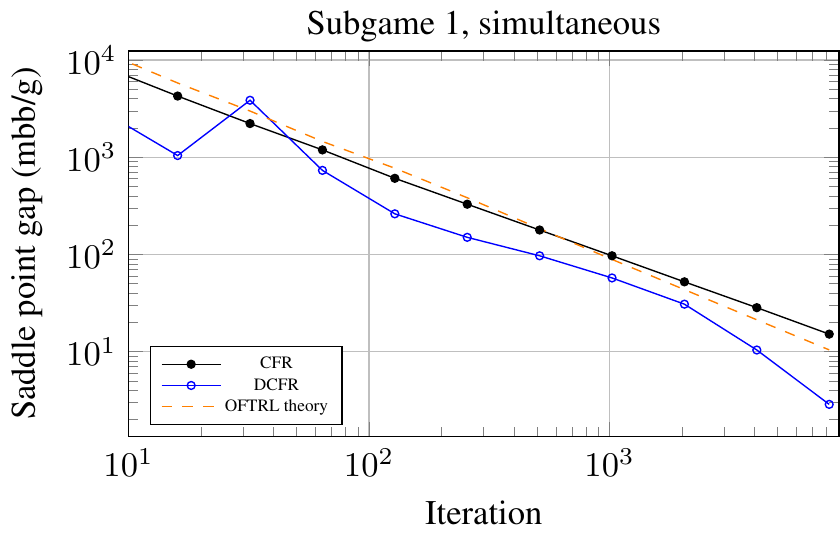}\hfill
  \includegraphics[width=.49\linewidth]{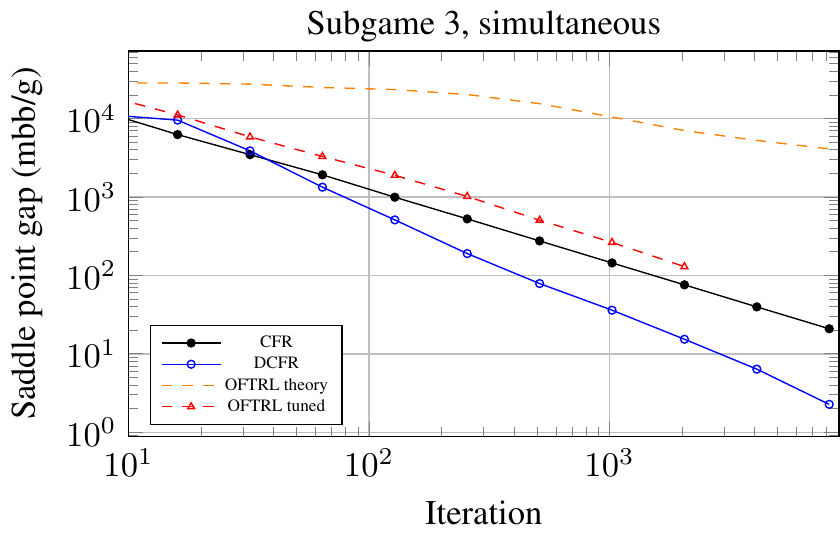}
  \caption{Convergence rate with iterations on the x-axis, and the exploitability in mbb. All algorithms use simultaneous updates.}
  \label{fig:plots simultaneous 1 3}
\end{figure}

Figure~\ref{fig:plots simultaneous 2 4} shows the results for simultaneous updates on subgames 2 and 4, while Figure~\ref{fig:plots simultaneous 1 3} for games 1 and 3. In the smaller subgames 3 and 4 we find that OFTRL with the stepsize set according to our theory outperforms CFR: in subgame 4 almost immediately and significantly, in subgame 3 only after roughly 800 iterations. In contrast to this we find that in the larger subgames 1 and 2 the OFTRL stepsize is much too conversative, and the algorithm barely starts to make progress within the number of iterations that we run. With a moderately-hand-tuned stepsize OFTRL beats CFR somewhat significantly. In all games DCFR performs better than OFTRL, and also significantly better than its theory predicts. This is not too surprising, as both {\cfrp} and the improved DCFR are known to significantly outperform their theoretical convergence rate in practice.

Figure~\ref{fig:plots alternating 2 4} shows the results for alternating updates on subgames 2 and 4, while games 1 and 3 are given in Figure~\ref{fig:plots alternating 1 3}. In the alternating-updates setting OFTRL performs worse relative to CFR and DCFR. In subgame 1 OFTRL with stepsizes set according to the theory slightly outperforms CFR, but in subgame 2 they have near-identical performance. In subgames 3 and 4 even the manually-tuned variant performs worse than CFR, although we suspect that it is possible to improve on this with a better choice of stepsize parameter. In the alternating setting DCFR performs significantly better than all other algorithms.

\begin{figure}[!h]
  \centering
  \includegraphics[width=.49\linewidth]{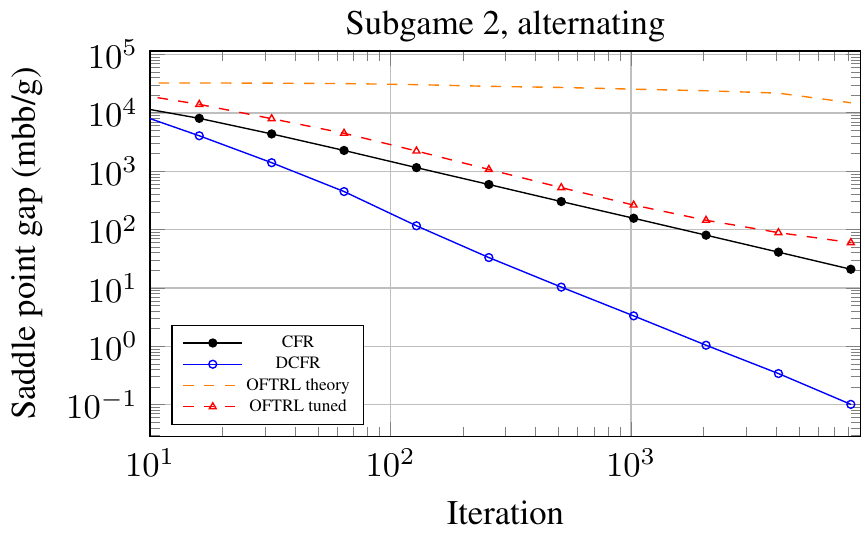}\hfill
  \includegraphics[width=.49\linewidth]{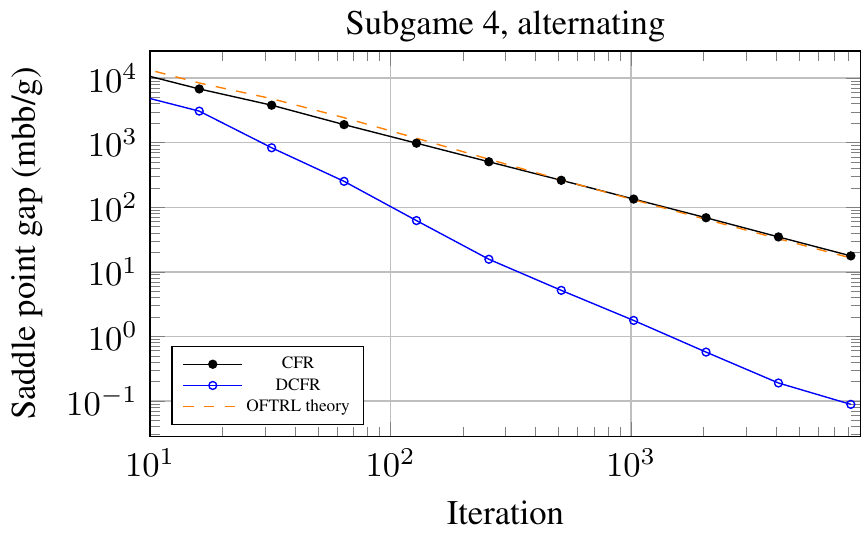}
  \caption{Convergence rate with iterations on the x-axis, and the exploitability in mbb. All algorithms use alternating updates.}
  \label{fig:plots alternating 2 4}
\end{figure}
\begin{figure}[H]
  \centering
  \includegraphics[width=.49\linewidth]{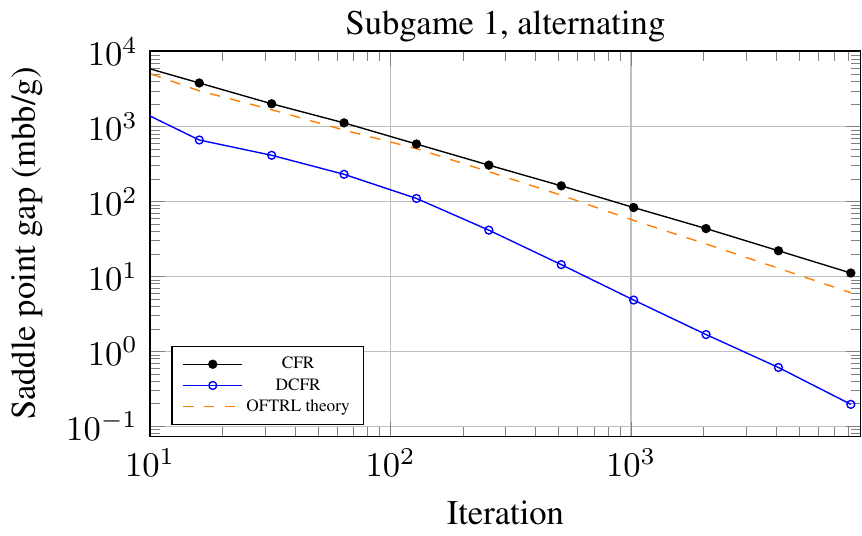}\hfill
  \includegraphics[width=.49\linewidth]{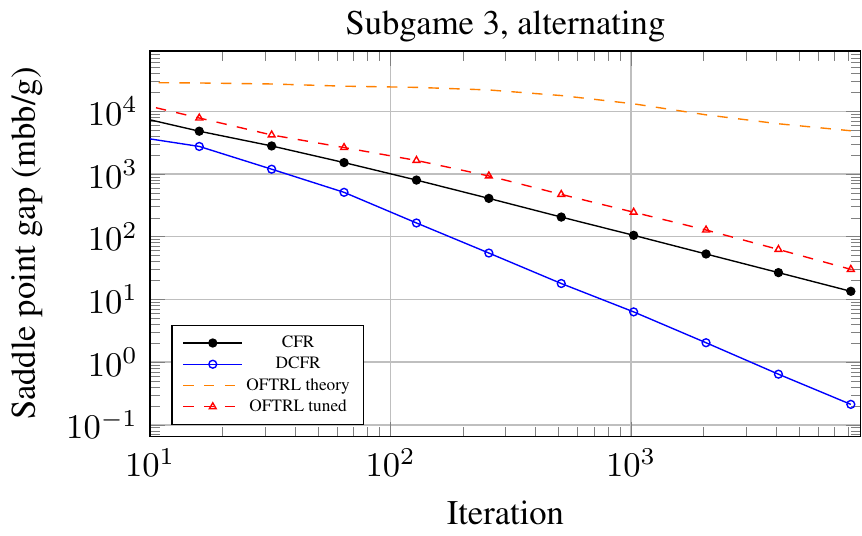}
  \caption{Convergence rate with iterations on the x-axis, and the exploitability in mbb. All algorithms use alternating updates.}
  \label{fig:plots alternating 1 3}
\end{figure}

\section{Conclusions}

We developed the first variant of CFR that converges at a rate better than $T^{-1/2}$. In particular we extend the ideas of predictability and stability for optimistic regret minimization on matrix games to the setting of EFGs. In doing so we showed that stable-predictive simplex regret minimizers can be aggregated to form a stable-predictive variant of CFR for sequential decision making, and we showed that this leads to a convergence rate of $O(T^{-3/4})$ for solving two-player zero-sum EFGs. Our result makes the first step towards reconciling the gap between the theoretical rate at which CFR converges, and the rate at which $O(T^{-1})$ first-order methods converge.

Experimentally we showed that our CFR variant can outperform CFR on some games,
but that the choice of stepsize is important, while we find that DCFR is faster
in practice. An important direction for future work is to find variants of our
algorithm that still satisfy the theoretical guarantee and perform even better
in practice.

\bibliography{dairefs}
\bibliographystyle{custom_arxiv}

\clearpage
\appendix
\section{Stable-predictivity of OFTRL}
We offer a proof of Theorem~\ref{eq:oftrl choice}.

  First, we introduce the following argmin-function:
  \begin{equation}\label{eq:argmin function}
    \tilde x : L \mapsto \argmin_{ x\in\cX} \left\{\left\langle  x, L \right\rangle + \frac{1}{\eta}R( x)\right\}.
  \end{equation}
  Furthermore, let $L^t \defeq \sum_{\tau=1}^{t} \ell^\tau$. With this notation, the decisions produced by OFTRL, as defined in~\eqref{eq:oftrl choice}, can be expressed as $x^t = \tilde x(L^{t-1} + m^t)$.

\textbf{Continuity of the argmin-function}. The first step in the proof is to study the continuity of the argmin-function $\tilde x$. Intuitively, the role of the regularizer $R$ is to \emph{smooth out} the linear objective function $\langle \cdot, L\rangle$. So, it seems only reasonable to expect that, the higher the constant that multiplies $R$, the less the argmin $\tilde x(L)$ is affected by small changes of $L$. In fact, the following holds:
\begin{lemma}\label{lem:x tilde lipschitz}
  The argmin-function $\tilde x$ is $\eta$-Lipschitz continuous with respect to the dual norm, that is
  \[
    \|\tilde x(L) - \tilde x(L')\| \le \eta \|L - L'\|_*.
  \]
\end{lemma}
\begin{proof}
  The variational inequality for the optimality of $\tilde x(L)$ implies
  \begin{equation}\label{eq:vi1}
    \left\langle L + \frac{1}{\eta}\nabla R(\tilde x(L)), \tilde x(L') - \tilde x(L) \right\rangle \ge 0.
  \end{equation}
  Symmetrically for $\tilde x(L')$, we find that
  \begin{equation}\label{eq:vi2}
    \left\langle L' + \frac{1}{\eta}R(\tilde x(L')), \tilde x(L) - \tilde x(L') \right\rangle \ge 0.
  \end{equation}
  Summing inequalities~\ref{eq:vi1} and~\ref{eq:vi2}, we obtain
  \begin{align*}
    &\frac{1}{\eta} \left\langle \nabla R(\tilde x(L)) - \nabla R(\tilde x(L')), \tilde x(L) - \tilde x(L') \right\rangle \le
    \left\langle L' - L, \tilde x(L) - \tilde x(L') \right\rangle.
  \end{align*}
  Using strong convexity of $R(\cdot)$ on the left-hand side and the generalized Cauchy-Schwarz inequality on the right-hand side, we obtain
  \begin{equation*}
    \frac{1}{\eta}\|\tilde x(L) - \tilde x(L')\|^2 \le \|\tilde x(L) - \tilde x(L')\|\,\|L - L'\|_\ast,
  \end{equation*}
  and dividing by $\|\tilde x(L) - \tilde x(L')\|$ we obtain the Lipschitz continuity of the argmin-function $\tilde x$.
\end{proof}
A direct consequence of Lemma~\ref{lem:x tilde lipschitz} is the following corollary, which measures the stability (small step size) of the decisions output by OFTRL:
\begin{corollary}At each time $t$, the iterates produced by OFTRL satisfy
$
  \|x^t - x^{t-1}\| \le 3\eta\Delta_\ell.
$
\end{corollary}
\begin{proof}
  \begin{align*}
    \|x^t - x^{t-1}\| &= \left\|\tilde x(L^{t-1} + m^t) - \tilde x(L^{t-2} + m^{t-1})\right\|\\
                      &\le {\eta}\|\ell^{t-1} + m^t - m^{t-1} \|_\ast \le {3\eta\Delta_\ell},
  \end{align*}
  where the first inequality holds by Lemma~\ref{lem:x tilde lipschitz} and the second one by definition of $\Delta_\ell$ and the triangle inequality.
\end{proof}
The rest of the proof, specifically the predictivity parameters $\alpha$ and $\beta$ of OFTRL follow directly from the proof of Theorem~19 in the appendix of~\citet{Syrgkanis15:Fast}.

\end{document}